\documentclass[a4paper]{article}

\usepackage[utf8]{inputenc} 
\usepackage[T1]{fontenc}
\usepackage{url}
\usepackage{ifthen}
\usepackage{cite}
\usepackage[cmex10]{amsmath} 
\usepackage{amssymb,color,graphicx,algorithm,algorithmic,amsthm}
\usepackage{array,arydshln,multirow}

\newtheorem{thm}{Theorem}
\newtheorem{lem}[thm]{Lemma}
\newtheorem{prop}[thm]{Proposition}
\newtheorem{cor}[thm]{Corollary}

\newtheorem{defini}[thm]{Definition}
\newenvironment{defi}{\begin{defini}\rm}{\end{defini}}

\newtheorem{rema}[thm]{Remark}

\newcommand{\field}[1]{\mathbb{#1}}
\newcommand{\F}{\field{F}}

\newcommand{\N}{\field{N}}
\newcommand{\R}{\field{R}}

\newcommand{\bs}{\boldsymbol}

\newcommand{\C}{\mathcal{C}}

\newcommand{\rk}{\mathrm{rank}}

\newcommand{\Di}{{D}_i}
\newcommand{\Ci}{{C}_i}
\newcommand{\Fq}{\mathbb{F}_q}
\newcommand{\Fm}{\mathbb{F}_{q^N}}
\newcommand{\E}{\mathcal E}

\begin{document}
\title{Random Construction of Partial MDS Codes}

\author{Alessandro Neri and Anna-Lena Horlemann-Trautmann}
\date{}

\maketitle

\begin{abstract}
  This work deals with partial MDS (PMDS) codes, a special class of locally repairable codes, used for distributed storage system. We first show that a known construction of these codes, using Gabidulin codes, can be extended to use any maximum rank distance code. Then we define a standard form for the generator matrices of PMDS codes and use this form to give an algebraic description of PMDS generator matrices. This implies that over a sufficiently large finite field a randomly chosen generator matrix in PMDS standard form generates a PMDS code with high probability. This also provides sufficient conditions on the field size for the existence of PMDS codes.
  
\end{abstract}


\section{Introduction}

In a distributed storage system a file $x\in \F_q^k$ (where $\F_q$ denotes the finite field of cardinality $q$), is encoded and stored  as some codeword $c\in \F_q^n$, over several storage nodes. Each of these nodes is assumed, for simplicity, to store exactly one coordinate of $c$. In case some of the nodes fail, we want to be able to recover the lost information using as little effort as possible.
The \emph{locality} of a code plays an important role in this context, and it denotes the number of nodes one has to contact for repairing a lost node. We call the set of nodes one has to contact if a given node fails, the locality group of that node.  
For this work we assume that the locality groups are distinct. \emph{Partial MDS codes} are \emph{maximally recoverable} codes in this setting, i.e., any erasure pattern that is information theoretically correctable is correctable with such a code. 

It is known that maximally recoverable codes in general, and PMDS codes in particular, exist for any locality configuration if the field size is large enough \cite{ch07}. Furthermore, some constructions of PMDS codes are known, e.g. \cite{bl13,bl14,bl16,ch15,go14}.

In this paper we describe a random construction of PMDS codes, by prescribing a generator matrix of the respective code in a specific form, which we will call \emph{PMDS standard form}. If we fill the non-prescribed coordinates of this generator matrix  with random values, by high probability, the resulting code is PMDS, if the underlying field size is large enough. We derive a lower bound on this probability (depending on the field size). This gives rise to a lower bound on the necessary field size for PMDS codes to exist. With some final adjustments on the random construction we get a lower bound that improves the one of \cite{ch07}.

\section{Preliminaries}


\subsection{PMDS codes}

Consider a distributed storage system with $m$ disjoint locality groups, where the $i$-th group is of size $n_i$ ($i=1,\dots,m$) and can correct any $r_i$ erasures. First we set the locality for the  code to be $\ell \in \mathbb N$. We can divide the coordinates of the code  into blocks of length $n_1, \dots, n_m$, where $n_i=\ell+r_i$, such that each block represents a locality group. 

We denote an MDS code of length $n$ and dimension $k$ by $[n,k]$-MDS code. We use the definition of PMDS codes given in \cite{ht17}, which generalizes the definition of \cite{bl13}.

\begin{defi}
 Let $\ell,m, r_1,\dots,r_m \in \mathbb N$. Define $n:=\sum_{i=1}^m (r_i+\ell)$ and let $C\subseteq \F_q^n$ be a linear code of dimension $k<n$
 with generator matrix
 \begin{equation}\label{eq:genmatrix}
G= \left( B_1 \mid \dots \mid B_m \right) \in \F_q^{k\times n}
\end{equation}
such that $B_i\in \F_q^{k\times ( r_i + \ell)}$. 
 Then $C$ is a $[n,k, \ell ; r_1,\dots,r_m]$-\emph{partial-MDS (PMDS) code} (with locality $\ell$) if 
 \begin{itemize}
 \item 
 for $i\in \{1,\dots,m\}$ the row space of $B_i$ is a $[r_i+\ell, \ell ]$-MDS code, and 
 \item
for any $r_i$ erasures in the $i$-th block ($i=1,\dots, m$), the remaining code (after puncturing the coordinates of the erasures) is a $[m\ell, k]$-MDS code.
\end{itemize}
\end{defi}

The erasure correction capability of PMDS codes is as follows:

\begin{lem}\label{lem:PMDScap}\cite[Lemma 3]{ht17}
A $[n,k, \ell ; r_1,\dots,r_m]$-PMDS code can correct any $r_i$ erasures in the $i$-th block (simultaneously) plus $s:= m\ell -k$ additional erasures anywhere in the code.
\end{lem}

We can see that the definition of PMDS codes given makes sense only for $k\geq \ell$. In case of equality, or in the case that $m=1$
 there exist only trivial PMDS codes, i.e. the only PMDS codes are MDS codes.


It was shown in \cite{ht17} that a code is a $[n,k, \ell; r_1,\dots,r_m]$-PMDS code if and only if it is maximally recoverable (for the respective locality group configuration). The same results had previously been shown in \cite[Lemma 4]{go14} for the case $r_1=r_2=\dots = r_m$.

Now we give a summary on  known results about PMDS codes.


\begin{prop}\cite{ch07}\label{prop4}
Maximally recoverable (MR) codes of length $n$ and dimension $k$ exist for any locality configuration over any finite field of size $q> \binom{n-1}{k-1}$.
\end{prop}

MR codes are PMDS codes for disjoint locality blocks. Therefore, Proposition \ref{prop4}  implies that PMDS codes exist for any set of parameters when the field size is large enough.

A construction of PMDS codes based on rank-metric and MDS codes was given in \cite{ca17}, when $r_1=r_2=\dots=r_m$. This gives the following existence result:
\begin{prop}\cite{ca17}
$[n,k,\ell; r,\dots,r]$-PMDS codes with $m$ locality blocks of the same length exist over a finite field of size  $q^{n-mr}$.
\end{prop}

Furthermore, some specific constructions of PMDS codes, for particular values of $s$ or of the $r_i$, are given in 
\cite{bl13,bl14,bl16,go14}. 

In particular, a general construction for PMDS codes with $s=1$  was given in \cite{ht17}.
This construction is based on the concatenation of several MDS codes as building blocks. 

\begin{prop}\cite[Corollary 14]{ht17}
\begin{enumerate}
\item
For any integers $m\geq 2 $ and $\ell,r_1,\dots,r_m \geq 1$ there exists a $[n,k=m\ell-1, \ell; r_1,\dots,r_m]$-PMDS code over any field $\F_q$ with $q\geq \max_i\{r_i\}+\ell$.

\item
If there exists $h\in \mathbb N$ such that $\ell\in \{3,2^h-1\}$ and $\max_i \{r_i\} +\ell = 2^h+1$, then there exists a  $[n,k=m\ell-1, \ell; r_1,\dots,r_m]$-PMDS code over $\F_q$ with $q=2^h= \max_i\{r_i\}+\ell -1$.

\end{enumerate}
\end{prop}

In \cite{ht17} the authors also show that this construction is basically the only one possible, i.e., every PMDS with $s=1$ is of this form, giving thus a characterization for this set of parameters. However, for $s\geq 2$ there is no characterization yet for PMDS codes.

\subsection{Zarisky topology over finite fields}

Let $\F$ be a field, and $\F[x_1,\ldots,x_N]$ be the polynomial ring over $\F$. Denote by $\bar{\F}$ the algebraic closure of $\F$. For a subset $S\subseteq \F[x_1,\ldots,x_N]$ we define the \emph{algebraic set} 
$$  V(S): = \{\bs \alpha \in \bar\F_q^r \mid f(\bs \alpha) = 0, \forall f \in S\} . $$

The \emph{Zariski topology} on $\bar{\F}^N$ is defined as the topology whose closed sets are the algebraic sets, while the complements of the Zariski-closed sets are the \emph{Zariski-open sets} \cite[Ch. I, Sec. 1]{ha13b}.

\begin{defi}
A subset  $A\subset\bar{\F}^N$ is called a \emph{generic set} if $A$ contains a non-empty Zariski-open set. 
\end{defi}

In classical geometry one studies the Zariski topology over the complex numbers. In this framework, a generic set inside $\mathbb C^N$ is dense and its complement is contained in an algebraic set of dimension at most $N-1$.

If one wants to consider generic sets restricted to a finite field $\F_q$, the situation is slightly different. Here, for a subset $T\subseteq\F_q^N$ one can always find a set of polynomials 
 $S\subseteq \F_q[x_1,\dots,x_N]$ such that 
$$   T=\{\bs \alpha \in \F_q^N \mid f(\bs \alpha) = 0, \forall f \in S\}.$$
and therefore the Zariski topology restricted to $\F_q^N$ is the discrete topology. This means that it is not useful to extend the notion of generic sets to finite fields since it would not give any  information.

However, given a set of polynomials $S\subseteq\F_q[x_1,\ldots,x_N]$,  we can  define the set of \emph{$\F_q$-rational points} as
$$ V(S;\F_{q}): = \{\bs \alpha \in  \F_{q}^N \mid f(\bs \alpha) = 0, \forall f \in S\} . $$
In this setting the Schwartz-Zippel Lemma 
 implies an analog result to the one of generic sets, as explained in the following.

\begin{lem}[Schwartz-Zippel]\cite[Lemma 1.1]{le98p}\label{lem:SZ}
  Let $f\in \F_q[x_1,\dots,x_N]$ be a non-zero polynomial of total
  degree $d \geq 0$. Let $T\subseteq \bar{\F}$ be a finite set and let $\alpha_1,   \dots, \alpha_N$ be selected at random independently and uniformly  from $T$. Then
  $$\Pr\big(f(\alpha_1,\ldots,\alpha_N)=0\big)\leq\frac{d}{|T|}. $$ 
\end{lem}

As a consequence of this result we have that, in case the size of $S$ and the total degrees of the polynomials in $S$ do not depend on the finite field, the proportion between the cardinality of $ V(S;\F_{q})$ and the cardinality of the whole space $\F_q^N$ goes to $0$ as $q$ grows. Vice versa, for growing $q$ the probability that a random point is in the complement of $ V(S;\F_{q})$ tends to $1$. This result will be crucial in Section \ref{sec:topprob} for our random construction of PMDS codes.

\subsection{Rank-metric codes}

We now give some known facts about rank-metric codes. Recall that $\Fm$ is isomorphic to $\Fq^N$ as an $\Fq$- vector space. From this it easily follows that $\Fm^n\cong \Fq^{N\times n}$. Then we can give the following definition.

\begin{defi}
  The \emph{rank distance} $d_R$ on $\F_q^{N\times n}$ is defined by
  \[d_R(U,V):= \rk(U-V) , \quad U,V \in \F_q^{N\times n}. \] 
  Analogously, if
  $\boldsymbol u,\boldsymbol v \in \Fm^n$, then $d_R(\boldsymbol u,\boldsymbol v)$ is the rank of the
  difference of the respective matrix representations in
  $\F_q^{N\times n}$.
\end{defi}

Observe that the definition of rank distance in the case of vectors in $\Fm^n$ does not depend on the choice of the basis. Moreover it can be shown that the function $d_R: \Fm^n \times \Fm^n \rightarrow \R_{\geq 0}$ is a metric.
\begin{defi}
 An \emph{$\Fm$-linear rank-metric code} $\C$ of length $n$ and dimension $k$ is a $k$-dimensional subspace of $\Fm^n$ equipped with the rank distance. 
 The \emph{minimum distance} of $\C$ is defined as 
$$d_R(\C):=\min\left\{d_R(\bs u,\bs v) \mid \bs u,\bs v\in \C, \bs u\neq \bs v\right\}.$$
\end{defi}


\begin{thm}[Singleton-like Bound]\label{th:SB}\cite[Theorem 1]{ro91}
Let $\C\subseteq \Fm^n$ be an $\Fm$-linear rank-metric code of dimension $k$. Then
$$d_R(\C)\leq n-k+1. $$
\end{thm}

\begin{defi}
Codes attaining the Singleton-like bound are called \emph{Maximum Rank Distance (MRD) Codes}.
\end{defi}

A necessary and sufficient condition for the existence of MRD codes is that $n\leq N$. In this framework, a characterization for $\Fm$-linear MRD codes in terms of their generator matrices was given 
in \cite[Corollary 2.12]{ho16}, which in turn is based on a  result given in
\cite{ga85}. For this we define the set
$$ \E_q(k,n):=\left\{E \in \F_q^{k\times n} \mid \rk(E)=k\right\}. $$

\begin{prop}[MRD criterion]\label{prop:MRDCrit}
  Let $G\in \F_{q^m}^{k\times n}$ be a generator matrix of a
  rank-metric code $\mathcal{C}\subseteq \F_{q^m}^n$. Then
  $\mathcal{C}$ is an MRD code if and only if
$$ \rk(GE^T) =k$$
for all $E\in \E_q(k,n)$.
\end{prop}

\section{General construction using rank metric codes}

In this section we generalize the construction given in \cite{ca17}. In that work the authors use Gabidulin codes in order to build $[n,k, \ell, r,\ldots,r]$-PMDS codes. We will show that this construction also works for different $r_i$, and that  Gabidulin codes can be replaced by any linear MRD codes.

 Fix $n, k,\ell, r_1,\ldots,r_m$, and let $\widetilde{G} \in \Fm^{k\times m\ell}$ be the generator matrix of a MRD code. For the existence of an MRD code we need $N\geq m\ell$. Moreover, for every $i=1,\ldots, m$, we consider a $[\ell+r_i,\ell]$-MDS code over $\Fq$ with generator matrix $M_i$, and define 
\begin{equation}\label{eq:MDSblock}
 M:=\left(\begin{array}{cccc} M_1 & 0 & \ldots & 0 \\
0 & M_2 &\ldots & 0 \\
\vdots &  & \ddots & \vdots \\
0 & \ldots & 0 & M_m
\end{array}\right) \in \Fq^{m\ell \times n}.
\end{equation}
We can now formulate our PMDS construction.

\begin{thm}\label{thm16}
Let $\widetilde{G} \in \Fm^{k\times m\ell}$ be the generator matrix of a MRD code and let $M$ be the matrix defined in (\ref{eq:MDSblock}). Then the matrix $\widetilde{G}M$ is a generator matrix for a $[n,k,\ell,r_1,\ldots,r_m]$-PMDS code over $\Fm$.
\end{thm}

\begin{proof}
 Let $G:=\widetilde{G}M$ 
 and let $S\in \mathcal T_{k,\ell}(G)$ be the submatrix obtained by selecting columns $h_1,\ldots,h_{k_j}$ from the $j$th block for $j=1,\ldots,m$, where $k_i\leq \ell$ and $k_1+\ldots+k_m=k$. 
 $S$ is equal to $\widetilde{G}\widetilde{M}$, where
$$\widetilde{M}=\left(\begin{array}{cccc} N_1 & 0 & \ldots & 0 \\
0 & N_2 &\ldots & 0 \\
\vdots &  & \ddots & \vdots \\
0 & \ldots & 0 & N_m
\end{array}\right), $$
and $N_j$ is the $\ell \times k_j$ submatrix of $M_j$ obtained by the respective selected columns. Since $M_i$ generates an $[\ell+r_i, \ell]$-MDS code, any $\ell$ columns of $M_i$ are linearly independent. Thus, $\rk(N_i)=k_i$ and $\rk(\widetilde{M})=k_1+\ldots+k_m=k$. By Proposition \ref{prop:MRDCrit} we have that $\det(\widetilde{G}\widetilde{M})\neq 0$, and we conclude the proof using Proposition~\ref{prop:PMDS}. 
\end{proof}

\begin{cor}\label{cor:exPMDS}
Let $m\geq 2 $ and $\ell,r_1,\dots,r_m \geq 1$, $k\geq \ell$  be positive integers. Then, for every prime $p$ and every positive integer $L\geq n_0m\ell$ there exists a $[n,k,\ell,r_1,\ldots,r_m]$-PMDS code over $\F_{p^L}$, where $$n_0=\min\{j \in \N \mid p^j \geq \ell + r_i-1, \mbox{ for } i=1,\ldots,m\}.$$
\end{cor}

\begin{proof}
A MRD code in $\Fm^{ m\ell}$ exists if $N\geq m\ell$. Suitable MDS codes over $\F_q$ for the matrix in \eqref{eq:MDSblock} exist if $q\geq \max\{\ell + r_i - 1\}$. 
The statement follows from Theorem \ref{thm16} with $q=p^{n_0}$ and $L=n_0N$.
\end{proof}

\section{Algebraic description of PMDS codes}

We will now define a standard form for generator matrices of PMDS codes. This standard form is the main tool for the random construction of PMDS codes.

\begin{thm}[PMDS standard form]\label{thm:stform}
Let $m\geq 2 $ and $s, \ell,r_1,\dots,r_m \geq 1$ and let $\C$ be a $[n,k=m\ell-s, \ell; r_1,\dots,r_m]$-PMDS code over a field $\F_q$. Then $\C$  has a generator matrix of the form
\begin{equation}\label{eq:stform}
{G}= \left({B}_1 \mid \dots \mid {B}_m \right), 
\end{equation}
where 
\begin{itemize}
\item ${B}_i=({C}_i \mid {D}_i)$, ${C}_i \in \F_q^{k \times \ell}$ and ${D}_i \in \F_q^{k\times r_i}$ for $i=1,\ldots, m$, and
\item
the submatrix ${G}_{C}=\left({C}_1 \mid \dots \mid {C}_m \right)$ is of the form
$${G}_C=\left[ I_{k} \mid A \right],$$
with $A$ being superregular.
\end{itemize}
\end{thm}

\begin{proof}
Let $\widetilde{G}$ be a generator matrix for $\C$ of the form (\ref{eq:genmatrix}), i.e.
$$\widetilde{G}= \left(\widetilde{B}_1 \mid \dots \mid \widetilde{B}_m \right). $$
Puncturing every block $\widetilde{B}_i$ in the last $r_i$ columns, we get that the submatrix
$\widetilde{G}_C$ is the generator matrix of a $[m\ell,k]$-MDS code. Operating on the rows of such a submatrix we can transform it to a matrix $G_C=\left[ I_{k} \mid A \right]$, with $A$ superregular. I.e., there exists an invertible matrix $P \in \mathrm{GL}_k(\Fq)$ such that $P\widetilde{G}_C=\left[ I_{k} \mid A \right]$, and therefore the matrix $G:=P\widetilde{G}$ is a generator matrix of $\C$ of the required form.
\end{proof}

We now consider the entries $a_{w,z}$ of $A$ as variables $x_{w,z}$ for $w=1,\ldots,k$ amd $z=1,\ldots,s$.
We know that the column space of ${D}_i$ is inside the column space of ${C}_i$, by the parameters of the block MDS codes. This means that every column in $\Di$ is a linear combination of the columns of $\Ci$. If we denote by $\Di^{(j)}$ the $j$th column of $\Di$, then
\begin{equation}\label{eq:columns}
\Di^{(j)}=\sum_{t=1}^{\ell}y_{t,i,j}\Ci^{(t)}
\end{equation}
for some $y_{t,i,j}$, which we also consider variable.
This way we can consider a $k\times n$ generator matrix as a matrix in $\F_q[x_{w,z}, y_{t,i,j}]^{k\times n}$ (where $\F_q[x_{w,z}, y_{t,i,j}]$ denotes the polynomial ring in all $x_{w,z}, y_{t,i,j}$). 

Let $R=\sum_{i=1}^m r_i$. We denote $\boldsymbol{\alpha} := (\alpha_{w,z})_{w,z}\in \F_q^{sk}$ and $\boldsymbol{\beta} := (\beta_{t,i,j})_{t,i,j} \in \F_q^{\ell R}$. If we replace the variables $x_{w,z}, y_{t,i,j}$ described above in a matrix in PMDS standard form by the values $\alpha_{w,z}, \beta_{t,i,j}$, we denote the corresponding generator matrix by
$$ G(\boldsymbol{\alpha},\boldsymbol{\beta}). $$
Analogously we will denote the variable form by $ G(\boldsymbol{x},\boldsymbol{y})$.

However, a general matrix of this form is not necessarily a generator matrix of a PMDS code for any values $\boldsymbol{\alpha},\boldsymbol{\beta}$. The following proposition shows what needs to be fulfilled to generate a PMDS code:

\begin{prop}\label{prop:PMDS}
 A matrix $G\in \F_q^{k\times n}$ generates a $[n,k=m\ell-s, \ell; r_1,\dots,r_m]$-PMDS code if and only if, every submatrix in the set
 $$\mathcal T_{k,\ell}(G):=\left\{S\in \F^{k \times k} \mid \begin{array}{l}S \mbox{ is a submatrix of } G \mbox{ with }\\ \mbox{ at most }  \ell \mbox{ columns per block } B_i \end{array}\right\}$$
 has non-zero determinant.
\end{prop}
\begin{proof}
This follows from the definition of PMDS, cf. also \cite{ht17}.
\end{proof}

The above results give an algebraic description of the generator matrix of a $[n,k=m\ell-s, \ell; r_1,\dots,r_m]$-PMDS code over $\F_q$, as follows. If we consider the variable form of a generator matrix $G$ as above, and the polynomial 
 \begin{equation}\label{eq:poly}
p(\boldsymbol{x},\boldsymbol{y}):= \mathrm{lcm} \{ \det S \mid S\in\mathcal T_{k,\ell}(G)\} \in  \F_q[x_{w,z}, y_{t,i,j}],
\end{equation}
 then, we have that $G(\boldsymbol{\alpha},\boldsymbol{\beta})$ generates a $[n,k=m\ell-s, \ell; r_1,\dots,r_m]$-PMDS code over $\F_q$ if and only if $p(\boldsymbol{\alpha},\boldsymbol{\beta})$ 
  is non-zero.

\section{Topological and probability results}\label{sec:topprob}

In this section we first deal with the algebraic description of the generator matrix of a PMDS code in the algebraic closure of the finite field where we want our code to be built.
After that, we analyze the probability that a code whose generator matrix is of the form $G(\boldsymbol{\alpha},\boldsymbol{\beta})$ is PMDS. Moreover, we also study the existence of PMDS codes for given parameters $n,k,\ell, s, r_1,\ldots,r_m$ and $R=\sum_{i=1}^m r_i$, giving sufficient conditions on the field size. Although for $s=1$ this problem was completely solved in \cite{ht17}, for  $s\geq 2$ this is still an open problem. 

We denote the set of valid entries for PMDS generator matrices over the algebraic closure of the finite field $\Fq$ by
$$\small{\mathcal A_{\mathrm{PMDS}}:=\left\{(\boldsymbol{\alpha},\boldsymbol{\beta}) \in \bar{\F}_q^{sk}\times \bar{\F}_q^{\ell R}  \mid \mathrm{row space}(G(\boldsymbol{\alpha},\boldsymbol{\beta})) \mbox{ is PMDS}\right\},} $$
Then the following result holds.

\begin{thm}\label{th:gen}
$\mathcal A_{\mathrm{PMDS}}$ is a generic set.
\end{thm}

\begin{proof} By Proposition \ref{prop:PMDS} we have that
$$\mathcal A_{\mathrm{PMDS}}:=\left\{(\boldsymbol{\alpha},\boldsymbol{\beta}) \in \bar{\F}_q^{sk}\times \bar{\F}_q^{\ell R} \mid p(\boldsymbol{\alpha},\boldsymbol{\beta})\neq 0 \right\}, $$
and therefore $\mathcal A_{\mathrm{PMDS}}$ is a Zariski open set. 

Concerning the non-emptiness, let $q=p^{t_0}$. From Corollary~\ref{cor:exPMDS} there exists  a $[n,k,\ell,r_1,\ldots,r_m]$-PMDS code $\C$ over $\F_{p^L}$, for some $L$ multiple of $t_0$. By Theorem \ref{thm:stform}, $\C$ has a generator matrix of the form $G(\boldsymbol{\alpha},\boldsymbol{\beta})$, therefore $(\boldsymbol{\alpha},\boldsymbol{\beta})\in \mathcal A_{\mathrm{PMDS}}$.
\end{proof}

This means that over the algebraic closure, by probability $1$, for randomly chosen $\boldsymbol{\alpha},\boldsymbol{\beta}$ the matrix $G(\boldsymbol{\alpha},\boldsymbol{\beta})$ generates a PMDS code. For underlying \emph{finite} fields, this implies that for growing field size this probability will tend to $1$. We now derive a probability formula depending on the field size.

 

We can easily observe that the entries of $G(\boldsymbol{x},\boldsymbol{y})$ are polynomials of total degree $0,1$ or $2$. In particular, if $t:=\lceil\frac{s}{\ell}\rceil$ and we write $G(\boldsymbol{x},\boldsymbol{y})$ as in \eqref{eq:stform}, then the entries of the blocks $\Di$ are polynomials of degree at most $1$ for $i=1,\ldots, m-t$, and of degree at most $2$ for the last $t$ blocks.

To estimate the degree of $p(\boldsymbol{x},\boldsymbol{y}) $ we need the following lemma.
\begin{lem}[Derivation of Vandermonde's identity]\label{lem:Vid}
$$\sum_{j=0}^r (r-j)\binom{m}{j}\binom{n}{r-j}=n\binom{m+n-1}{r-1}.$$
\end{lem}


\begin{lem}\label{lem:deg} 
The total degree of the polynomial $p(\boldsymbol{x},\boldsymbol{y})$, defined as in (\ref{eq:poly}), satisfies the inequality
$$\deg p(\boldsymbol{x},\boldsymbol{y}) \leq 2(n-k)\binom{n-1}{k-1}.$$
\end{lem}

\begin{proof}
 It holds that 
$$\mathcal T_{k,\ell}(G) \subset \mathcal M_k(G):=\{ S \in \Fq^{k\times k} \mid S \mbox{ is a submatrix of } G\},$$
hence the polynomial $p(\boldsymbol{x},\boldsymbol{y})$ divides the polynomial 
$$q(\boldsymbol{x},\boldsymbol{y}):=\mathrm{lcm}\{\det S \mid S\in \mathcal M_k(G)\}.$$
Observe that the entries of the first $k$ columns of the submatrix ${G}_C$ have degree 0. Let $t:=\lceil\frac{s}{\ell}\rceil$. Then the entries of the columns corresponding to the blocks $\Di$ for $i=1,\ldots, m-t$ have degree at most $1$, as well as the last $m\ell -k$ columns of ${G}_C$. Finally, the columns of the blocks $\Di$ for $i=m-t+1,\ldots, m$, have degree at most $2$. 
In particular, all the entries of the blocks $\Di$ and the last $m\ell -k$ columns of ${G}_C$ have degree at most $2$. Therefore, 
\begin{align*}
\deg q(\boldsymbol{x},\boldsymbol{y}) &\leq \sum_{S \in \mathcal M_k(G)} \deg \det S \\
&\leq \sum_{j_0=0}^k2(k-j_0) \binom{k}{j_0}\binom{n-k}{k-j_0}\\
&=2(n-k)\binom{n-1}{k-1}
\end{align*}
where the last equality follows from Lemma \ref{lem:Vid}. Since $\deg p(\boldsymbol{x},\boldsymbol{y}) \leq \deg q(\boldsymbol{x},\boldsymbol{y})$ we conclude the proof.
\end{proof}

We can now formulate a lower bound for the probability that a randomly chosen generator matrix in PMDS standard form generates a PMDS code over a finite field $\F_q$:

\begin{thm}\label{thm:Prob}
Let the entries of $\boldsymbol{\alpha}$ and $\boldsymbol{\beta}$ be  uniformly and independently chosen at random in $\Fq$. Then
 $$\mathrm{Pr}\{\mathrm{row space}(G(\boldsymbol{\alpha},\boldsymbol{\beta})) \mbox{ is PMDS  }\}\geq 1- \frac{2(n-k)\binom{n-1}{k-1}}{q}.$$
\end{thm}

\begin{proof}
 We have 
\begin{align*}
 &\mathrm{Pr}\{\mathrm{row space}(G(\boldsymbol{\alpha},\boldsymbol{\beta})) \mbox{ is PMDS }\} \\
=&\mathrm{Pr}\{(\boldsymbol{\alpha},\boldsymbol{\beta}) \notin V(p(\boldsymbol{x},\boldsymbol{y});\Fq)\}  \\
=&1-\mathrm{Pr}\{p(\boldsymbol{\alpha},\boldsymbol{\beta})=0\} \\
\geq& 1-\frac{\deg p(\boldsymbol{x},\boldsymbol{y})}{q}\geq 1 - \frac{2(n-k)\binom{n-1}{k-1}}{q},
\end{align*}

where the last two inequalities follow from Lemmas \ref{lem:SZ} and \ref{lem:deg}, respectively.
\end{proof}

From this we can deduce an existence result for PMDS codes over finite fields of a given minimal size.

\begin{cor}
 If $q>2(n-k)\binom{n-1}{k-1}$ then there exists a $[n,k,\ell,r_1,\ldots,r_m]$-PMDS code over the finite field $\Fq$.
\end{cor}

One notices that this is not an improvement over the known existence result from Proposition \ref{prop4}.

However, we
can improve the above result, considering a step-by-step construction. We will again consider a generator matrix in PMDS standard form as in \eqref{eq:stform}. We start with an $[m\ell,k]$-MDS code over a finite field $\Fq$ and write its generator matrix as $(C_1 \mid \dots \mid C_m)$. For this purpose it is sufficient that $q\geq m\ell -1$. Then we construct the first column $D_1^{(1)}$ of the block $D_1$ as in (\ref{eq:columns}). Every entry will be a degree $1$ polynomial in the variables $y_{t,1,1}$ for $t=1,\ldots,\ell$. Imposing that every $k\times k$ minor of 
\begin{equation*}G':=\left( {C}_1\mid {D}_1^{(1)} \mid {C}_2\mid \ldots \mid {C}_m \right)\end{equation*}
 is non-zero, we get the condition $p'(y_{t,1,1})\neq 0$, where 
\begin{equation*}p'(y_{t,1,1})=\mathrm{lcm}\{ \det S \mid S \in \mathcal T_{k,\ell}(G')\}.\end{equation*}
 Using Lemma \ref{lem:SZ}, we obtain
\begin{equation*}\mathrm{Pr}\{p'(\beta_{t,1,1})= 0\}\leq \frac{\deg p'}{q}.\end{equation*}
In this situation $\deg p' \leq \binom{m\ell}{k-1}$, therefore for $q> \binom{m\ell}{k-1}$ we have that there exists at least one evaluation of $p'$ that is non-zero and such that $G'(\beta_{t,1,1})$ generates a $[n,k,\ell, 1,0,\ldots,0]$-PMDS code.

Repeating this construction step by step, we get a $[n,k,\ell,r_1,\ldots,r_{m-1},r_{m}-1]$-PMDS code. From that code we build the last column ${D}_m^{(r_m)}$ of the block ${D}_m$ again as in (\ref{eq:columns}): 
\begin{equation*}{D}_m^{(r_m)}=\sum_{t=1}^{\ell}y_{t,m,r_m}{C}_m^{(t)}.\end{equation*}
In the end we get the matrix

\begin{equation}\label{eq:G}
{G}(y_{t,m,r_m})=({C}_1\mid {D}_1 \mid \ldots \mid {C}_m \mid {D}_m ),
\end{equation}
where the matrix ${G}(y_{t,m,r_m})$ without the last column generates a $[n,k,\ell,r_1,\ldots,r_{m-1},r_{m}-1]$-PMDS code, and the entries of the last column are polynomials of total degree at most $1$ in the variables $y_{t,m,r_m}$, for $t=1,\ldots,\ell$. 

\begin{defi}
Let $m, n, k, n_1,\ldots,n_m, f_1,\ldots,f_m$ be  positive integers such that $n=\sum_i n_i$. Let $N_0:=0$, $N_i:=\sum_{j=1}^i n_j$ and $J_i=\{N_{i-1}+1,\ldots,N_i\}$ for $i=1,\ldots,m-1$. We define the set
\begin{equation*}
\mathcal M(k;n_1\ldots,n_m;f_1,\ldots,f_m)
=\> \left\{I \subset \{1,\ldots,n\}\mid |I|=k, |I\cap J_i|\leq f_i\right\}
\end{equation*}
and $M(k;n_1\ldots,n_m;f_1,\ldots,f_m)$ 
 as its cardinality.
\end{defi}

\begin{prop}
 Let ${G}(y_{t,m,r_m})$ be as in (\ref{eq:G}). The total degree of the  polynomial 
\begin{equation*}\tilde{p}(y_{t,m,r_m}):=\mathrm{lcm}\{\det S \mid S\in \mathcal T_{k,\ell}({G}(y_{t,m,r_m}))\}\end{equation*} is less or equal to
\begin{equation*} M(k-1;\ell+r_1\ldots,\ell+r_{m-1},\ell+r_m-1;\ell,\ldots,\ell,\ell-1) =: M^*.\end{equation*}
\end{prop}

\begin{proof}
The polynomial $\tilde{p}(y_{t,m,r_m})$ has degree less or equal to $\sum \deg \det S$. By assumption all the determinants $\det S$ for $S$ not containing the last column are non zero elements in $\Fq$. The only polynomials with degree $1$ are the determinants of $k\times k$ submatrices involving the last column, and they are exactly $M^*$ many.
\end{proof}

\begin{cor}\label{cor25}
If $q>M^*$ then there exists a $[n,k,\ell,r_1,\ldots,r_m]$-PMDS code over the finite field $\Fq$.
\end{cor}

To our knowledge there is no closed formula for $M^*$. However, it is easy to see that $M^* \leq \binom{n-1}{k-1}$ and that the inequality is strict if any of the conditions $|I\cap J_i|\leq f_i$ is non-empty. Hence, Corollary \ref{cor25} improves upon Proposition~\ref{prop4}.

\section{Conclusion}

We gave a generalization of a known PMDS code construction based on rank-metric codes. Furthermore, we investigated a random construction of PMDS codes by prescribing a PMDS standard form. We derived a lower bound on the probability that a randomly filled matrix in PMDS standard form generates a PMDS code. This probability implies a lower bound on the field size needed for such codes to exist. In the end we gave a step-by-step construction of such a generator matrix to improve this lower bound on the necessary field size.




\bibliographystyle{IEEEtran}
\bibliography{PMDS_stuff}

\end{document}